\documentclass[letterpaper, 10 pt, conference]{ieeeconf}  

\IEEEoverridecommandlockouts                              
\usepackage{cite}
\usepackage{amsmath,amssymb,amsfonts,amsthm,color,float}
\usepackage{algorithmic}
\usepackage{graphicx}
\usepackage{textcomp}
 \usepackage{calc}
 \usepackage{tikz}
\usepackage{cases}
\usepackage{verbatim}
\usepackage{hyperref}
\usepackage{xcolor, soul}
\usepackage{nccmath}
\usepackage{mathtools, nccmath}

\usetikzlibrary{arrows}
\usetikzlibrary{shapes}
 \usetikzlibrary{calc}
\usetikzlibrary{decorations.pathreplacing}
\usetikzlibrary{arrows,positioning}

\theoremstyle{plain}
\newtheorem{thm}{Theorem}

\newtheorem{lemma}{Lemma}

\theoremstyle{definition}
\newtheorem{defn}{Definition}

\newtheorem{ex}{Example}
\newtheorem*{rem*}{Remark}

\theoremstyle{remark}

\newcommand{\R}{\mathbb{R}}
\newcommand{\N}{\mathbb{N}}
\renewcommand{\S}{\mathcal{S}}

\newcommand{\ac}{\mathcal{A}} 
\newcommand{\rr}{\mathcal{R}} 
\newcommand{\p}{\mathcal{I}} 
\newcommand{\W}{\mathrm{W}} 
\newcommand{\U}{\mathrm{U}_i} 
\newcommand{\G}{\mathrm{G}} 
\newcommand{\Gset}{\mathcal{G}_{f, n}} 
\newcommand{\br}{\text{BR}_i} 
\renewcommand{\ne}{a^{\mathrm{ne}}} 
\newcommand{\kend}{\mathrm{end}} 
\newcommand{\NE}{\mathrm{NE}} 
\newcommand{\pob}{\mathrm{PoB}} 
\newcommand{\poa}{\mathrm{PoA}} 

\newcommand{\ba}{\mathbf{a}}  
\newcommand{\emp}{a^{\varnothing}}  
\newcommand{\opt}{a^{\mathrm{opt}}}  
\newcommand{\abr}{a^{\mathrm{br}}}  
\newcommand{\pobset}{\pob(f, n, k)}
\newcommand{\poaset}{\poa(f, n)}
\newcommand{\ipoa}{\mathcal{X}}
\newcommand{\babr}{\mathbf{a}_{\mathrm{br}}}  
\newcommand{\baopt}{\mathbf{a}_{\mathrm{opt}}}  

\begin{document}
\title{Balancing Asymptotic and Transient Efficiency Guarantees\\ in Set Covering Games}

\author{Rohit Konda, Rahul Chandan, David Grimsman and Jason R. Marden
\thanks{R. Konda (\texttt{rkonda@ucsb.edu}), R. Chandan, D. Grimsman and J. R. Marden are with the Department of Electrical and Computer Engineering at the University of California, Santa Barbara, CA. This work is supported by \texttt{ONR Grant \#N00014-20-1-2359} and \texttt{AFOSR Grant \#FA9550-20-1-0054}.}}

\maketitle
\thispagestyle{empty}

\begin{abstract}
Game theoretic approaches have gained traction as robust methodologies for designing distributed local algorithms that induce a desired overall system configuration in multi-agent settings. However, much of the emphasis in these approaches is on providing asymptotic guarantees on the performance of a network of agents, and there is a gap in the study of efficiency guarantees along transients of these distributed algorithms. Therefore, in this paper, we study the transient efficiency guarantees of a natural game-theoretic algorithm in the class of set covering games, which have been used to model a variety of applications. Our main results characterize the optimal utility design that maximizes the guaranteed efficiency along the transient of the natural dynamics. Furthermore, we characterize the Pareto-optimal frontier with regards to guaranteed efficiency in the transient and the asymptote under a class of game-theoretic designs. Surprisingly, we show that there exists an extreme trade-off between the long-term and short-term guarantees in that an asymptotically optimal game-theoretic design can perform arbitrarily bad in the transient.  
\end{abstract}

\section{Introduction}
\label{sec:int}

Recently, game theory has emerged as a significant framework in the study of multi-agent systems: the agents in the system are modeled as agents in a game, each with its own actions and individual rewards that are functions of the joint action of the other agents. Some prominent examples where these models have found success include wireless communication networks \cite{han2012game}, UAV swarm task allocation \cite{roldan2018should}, news subscription services \cite{hsu2020information}, vaccinations during an epidemic \cite{hota2019game}, security systems \cite{9028855}, facility location \cite{1181966}, coordinating the charging of electric vehicles \cite{martinez2020decentralized}, and national defense \cite{lee2020perimeter}, among others. Indeed, game theory has played a significant role in many applications and across various disciplines.

A central focus in using game-theoretic models to control multi-agent systems is the design of local objective (or reward) functions for the agents that result in a desirable emergent system-level outcome, as measured by a system welfare function \cite{marden2015game}. Generally the emergent outcome is assumed to be a coordinated decision which is a Nash equilibrium of the game, wherein each agent cannot improve its own objective function by switching actions. A significant research effort has been devoted to designing local objective functions that ensure good performance of Nash equilibria. For example, \cite{gairing2009covering} shows that for the class of set covering games, there exists a set of local objective functions which guarantee that the performance of any Nash equilibrium is within a multiplicative factor of $1 -\frac{1}{e}$ of the optimal decision set. Furthermore, no other polynomial time algorithm can perform better \cite{feige1998threshold}. 



Rather than focusing solely on the system welfare at the Nash equilibria, this work studies the welfare of decision sets along the transient path toward equilibrium. This may be important when the relevant system parameters evolve on faster time scales, agent modalities vary with time etc. The class of \emph{potential games} \cite{monderer1996potential} is particularly relevant to this goal as the class of \emph{best response/reply processes} and its many variants (see, e.g., \cite{young1993evolution}) are known to converge to a Nash equilibrium asymptotically from any initial configuration. Best reply processes are intimately tied with the development of game theory, present in Cournot processes, Nash's seminal paper \cite{nash1950equilibrium}, as well as stable and evolutionary equilibrium \cite{sandholm2020evolutionary}. Efficiency guarantees along these best response processes are much less studied than the asymptotic guarantees that come with Nash equilibrium. In this work, we restrict our attention to an important subclass of potential games known as \emph{set covering games}, in which agents with limited capabilities try to coordinate to cover as large of an area as possible. For set covering games, we characterize bounds on the efficiency for best response processes, and compare utility designs that maximize either the short- or long-efficiency.

While few works study the guarantees that can arise through best response processes, we highlight a subset of important works that deal with analyzing the emergent behavior of best response dynamics. The authors in \cite{christodoulou2006convergence, bilo2009performances} analyze lower bounds for the efficiency guarantees resulting after a single best response round in linear congestion and cut games. These guarantees were verified through a linear programming approach in \cite{bilo2018unifying} and extended to quadratic and cubic congestion games. Best response dynamics are also studied in \cite{mirrokni2004convergence}, where it was found that the efficiency resulting from best response dynamics can be arbitrarily bad for non-potential games. The speed of convergence to efficient states is studied in \cite{fanelli2008speed}. While similar in spirit, our paper extends the previous literature by moving from characterization results of prespecified utility designs to instead identifying optimal game-theoretic designs that maximize the efficiency guarantees along a best response process.

The contributions of this paper are all with respect to the class of set covering games and are as follows: 
\begin{enumerate}
    \item We identify the local objective that optimizes the transient efficiency guarantee along a best response dynamic;
    \item We identify a trade-off between transient and asymptotic guarantees, and provide explicit characterizations of the Pareto frontier and Pareto-optimal agent objective functions.
\end{enumerate}

The rest of the paper is organized as follows: In Section~\ref{sec:back}, we introduce the necessary notation and definitions for best response dynamics, as well as for set covering games. We then outline our main results in Section \ref{sec:main}. Section \ref{sec:proofthm} contains the proof of the main result. We present a simulation study in Section \ref{sec:sim} to examine how our findings translate to practice. Finally, we conclude in Section \ref{sec:conc}.

\section{Model}
\label{sec:back}

\noindent \emph{Notation.} We use $\R$ and $\N$ to denote the sets of real and natural numbers respectively. Given a set $\S$, $|\S|$ represents its cardinality. $\mathbf{1}_S$ describes the indicator function for set $S$ ($\mathbf{1}_S(e) = 1$ if $e \in S$, $0$ otherwise).

\subsection{Set Covering Games and Best Response Dynamics}
\label{subsec:PoAB}

We present all of our results in the context of \emph{set covering games} \cite{gairing2009covering}, a well-studied class of potential games. We first recall some definitions to outline our game-theoretic model, which is used to describe the given class of distributed scenarios. Set covering games are characterized by a finite set of resources $\rr$ that can be possibly covered by a set of $n$ agents $\p = \{1, \dots, n\}$. Each resource $r$ has an associated value $v_r\geq 0$ determining its relative importance to the system. Each agent $i$ has a finite action set $\ac_i = \{a^1_i, \dots, a_i^c\} \subset 2^{\rr}$, representing the sets of resources that each agent can opt to select. The performance of a joint action $a = (a_1, \dots, a_n) \in \ac = \ac_1 \times \cdots \times \ac_n$ comprised of every agent's actions is evaluated by a system-level objective function $\W: \ac \to \R_{\geq 0}$, defined as the total value covered by the agents, i.e.,
\begin{equation}
    \W(a) = \sum_{r \in \bigcup_i a_i}{v_r}.
\end{equation}
The goal is to coordinate the agents to a joint action $\opt$ that maximizes this welfare $\W(a)$. Thus, each agent $i\in\p$ is endowed with its own utility function $\U : \ac \to \R$ that dictates its preferences among the joint actions. We use $a_{-i} = (a_1, \dots, a_{i-1}, a_{i+1}, \dots a_n)$ to denote the joint action without the action of agent $i$. A given set covering game $\G$ can be summarized as the tuple $\G  = (\p, \ac, \W, \{\U\}_{i \in \p})$. We consider the utility functions of the form
\begin{equation}
    \U(a_i, a_{-i}) = \sum_{r \in a_i} v_r \cdot f(|a|_r),
\end{equation}
where $|a|_r$ is the number of agents that choose resource $r$ in action profile $a$, and the \emph{utility rule} $f : \{1, \dots, n\} \to \R_{\geq 0 }$ defines the resource independent agent utility determined by $|a|_r$. This utility rule will be our design choice for defining the agents preferences towards each of their decisions. To motivate our results, we provide an example application of set covering games below, which represents just one among many \cite{hochbaum1985approximation, mandal2021covering, nash1977optimal}.

\begin{ex}[Wireless Sensor Coverage \cite{huang2005coverage}]
Consider a group of sensors $s_1, s_2, \dots, s_n$ that can sense a particular region depending on the orientation, physical placement of the sensor, etc. The choice of these parameters constitute the different possible actions $\ac_i$ for each sensor. We partition the entire space into disjoint regions $r\in\rr$ and each region $r$ is attributed a weight $v_r\geq 0$ to signify the importance of covering it. As a whole, the set of sensors wish to maximize the total weight of the regions covered. Using a game-theoretic model of this scenario allows us to abstract away much of the complexity that arises from having to consider specific sensor types and their respective capabilities, in contrast to \cite{huang2005coverage}. In this scenario, the results in the paper allow us to study guarantees of performance of local coordination algorithms that the sensors may employ. We further note that short-term performance guarantees may also be desirable especially when there is some drift in the parameters of the game (e.g., in $v_r$ or $\ac_i$).
\end{ex}

Next, we introduce some notation relevant to the study of best reply processes. For a given joint action $\bar{a} \in \ac$, we say the action $\hat{a}_i$ is a \emph{best response} for agent $i$ if 
\begin{equation}
\hat{a}_i \in \ \br(\bar{a}_{-i}) = \arg \max _ {a_i \in \ac_i} \U(a_i, \bar{a}_{-i}). 
\end{equation}
Note that $\br(\bar{a}_{-i})$ may be set valued. In this fashion, the underlying discrete-time dynamics of the best response process is
\begin{align}
    a_i[m+1] &\in \br(a_{-i}[m]) \text{ for some agent } i \\
    a_{-i}[m+1] &= a_{-i}[m] \nonumber,
\end{align}
where at each step, a single agent acts in a best response to the current joint action profile. It is important to stress that the fixed points of this best response process are Nash equilibria of the game, i.e., a joint action $\ne \in \ac$ such that
\begin{equation}
    \ne_i \in \br(\ne_{-i}) \text{ for all agents } i \in \p.
\end{equation}
In the case of potential games \cite{monderer1996potential}, the set of Nash equilibrium $\NE$ is also globally attractive with respect to the best response dynamics. A potential game is defined as any game for which there exists a potential function $\phi$ where for all agents $i$ and pairs of actions $a, a'$,
\begin{equation}
    \phi(a_i, a_{-i}) - \phi(a_i', a_{-i}) = \U(a_i, a_{-i}) - \U(a_i', a_{-i}).
\end{equation}
In these games, the potential function decreases monotonically along the best response process in a similar manner to a Lyapunov function. For non-potential games, there is a possibility that the best response process ends up in a cycle \cite{goemans2005sink}, but this falls outside the scope of this paper. 

In this work we assume that the best response process is such that the agents are ordered from $1$ to $n$ and perform successive best responses in that order. While other best response processes have been previously studied (see, e.g., $k$ covering walks \cite{christodoulou2006convergence}), we limit our analysis to this specific process for tractability. We also assume that each agent has a null action $\emp_i$ available to them representing their non-participation in the game $\G$.

\begin{defn}
Given a game $\G$ with $n$ agents, a \emph{$k$-best response round} is a $n \cdot k$ step best reply process with an initial state $a[0] = \emp$ and the corresponding dynamics
\begin{align}
    a_i[m+1] &\in \br(a_{-i}[m]) \text{ for agent } i = (m+1 \text{ mod } n) \nonumber \\
    a_{-i}[m+1] &= a_{-i}[m].
\end{align}
\end{defn}
This process results in a set of possible action trajectories $\Sigma^k \ni \sigma^k = (a[0], a[1] \dots, a[kn - 1 ], a[kn])$ with a set of possible end action states $\kend(k) = \{\sigma^k(kn) : \sigma^k \in \Sigma^k\}$ resulting after running for $k \cdot n$ time steps.

\subsection{Price of Best Response and Price of Anarchy}

The primary focus of this paper is to analyze guarantees on the efficiency of possible end actions $\kend(k)$ after a $k$-best response round with respect to the system welfare $\W$. In this regard, we introduce the \emph{price of best response} metric to evaluate the worst case ratio between the welfare of the outcome of the $k$-best response round and the optimal welfare in a game.

\begin{equation}
\pob(\G, k) = \frac{\min_{a \in \kend(k)}{\W(a)}}{\max_{a \in \mathcal{A}}{\W(a)}}.
\end{equation}

Our study of the price of best response will be in direct comparison with the standard game-theoretic metric of \emph{price of anarchy}, which is defined as

\begin{equation}
\poa(\G) = \frac{\min_{a \in \NE}{\W(a)}}{\max_{a \in \mathcal{A}}{\W(a)}}.
\end{equation}
The price of anarchy is a well established metric, with a number of results on its characterization, complexity, and design \cite{paccagnan2018utility}. We frame price of anarchy as a backdrop to many of our results, where Nash equilibria are the set of end states that are possible, following a $k$-best response round in the limit as $k \to \infty$. Under this notation, we can analyze the transient guarantees, which have been relatively unexplored, in these distributed frameworks in comparison to the asymptotic setting.


To isolate the efficiency analysis over the utility rule $f$ and the number of agents $n$, we consider the price of best response with respect to the class of all set covering games $\Gset$ with a fixed utility rule $f$ and $n$ number of agents, and any resource set $\rr$ and corresponding values $v_r$ and set of agent actions $\ac$:

\begin{equation}
\pobset = \inf_{\G \in \Gset} \pob(\G, k).
\end{equation}
Similarly, we define the price of anarchy over $\Gset$ as

\begin{equation}
\poaset = \inf_{\G \in \Gset} \poa(\G).
\end{equation}
When appropriate, we may also remove the dependence on the number of agents by denoting $\pob(f, k) = \inf_{n \in \N} \pob(f, n, k)$ (and likewise for $\poa$) to take the worst case guarantees across any number of agents.

We are interested in the transient and asymptotic system guarantees that can result from a given utility rule. Therefore we bring to attention two well-motivated utility rules to highlight the spectrum of guarantees that are possible. The first utility rule $f_{\poa} \in \arg \max_f \poa(f)$, introduced in \cite{gairing2009covering}, attains the optimal price of anarchy guarantee for the class of set covering games and is defined as
\begin{equation}
    f_{\poa}(j) = \sum_{\ell = j}^{\infty}{\frac{(j-1)!}{\ell!(e - 1)}}.
\end{equation}
The second rule is the \emph{marginal contribution} (MC) rule. The MC rule is a well studied \cite{vetta2002nash} utility rule and is defined as
\begin{equation}
    \label{eq:MCdefgen}
    f_{MC}(j) = \left\{\begin{array}{lr}
        1, & \text{for } j = 1\\
        0, & \text{otherwise}\\
        \end{array}\right\}.
\end{equation}
Note that the induced utility under this mechanism is $\U(a_i, a_{-i}) = \W(a_i, a_{-i}) - \W(\emp_i, a_{-i})$, where the utility for agent $i$ is the added welfare from playing its null action $\emp_i$ to playing $a_i$. Using $f_{MC}$ ensures that all Nash equilibria of the game are local optima of the welfare function. However, it is known that the MC rule has sub-optimal asymptotic guarantees as $\poa(f_{MC}) = \frac{1}{2} < 1 - \frac{1}{e} = \poa(f_{\poa})$ \cite{ramaswamy2019multiagent}. Thus, under the standard tools of asymptotic analysis, $f_{\poa}$ has been considered to be `optimal' in inducing desirable system behavior. Utilizing the notion of $\pob$, we inquire if the favorable characteristics of $f_{\poa}$ translate to good short-term behavior of the system. Intuitively, since $f_\poa$ guarantees optimal long-term performance, we expect the same in the short-term as well.
This discussion motivates the following three questions:
\begin{enumerate}
    \item[(a)] For a given $k\geq 1$, is $\pob(f_\poa, n, k)$ always greater than $\pob(f_{MC}, n, k)$?.
    \item[(b)] If not, then for how many rounds does the best response algorithm under $f_\poa$ have to be run for the guarantees on the end state performance to surpass the asymptotic performance of marginal contribution? In other words, is there $k\geq 1$ for which $\pob(f_{\poa}, n, k) \geq \poa(f_{MC}, n) = \frac{1}{2}$?
    \item[(c)] For a given $k\geq 1$, what is the utility mechanism that maximizes the end state performance after running the best response algorithm for $k$ rounds, i.e. what is $\arg \max_f \pob(f, n, k)$?
\end{enumerate}
In the next section, we answer these questions in detail.

\section{Main Results}
\label{sec:main}
In this section, we outline our main results for price of best response guarantees. In this setting, utility functions are designed by carefully selecting the utility rule $f$ and improvement in transient system level behavior corresponds with a higher price of best response. Hence, we characterize the optimal $f$ that achieves the highest possible $\pob$. We formally state the optimality of the MC rule for the class of set covering games below.

\begin{thm}
\label{thm:MCopt}
Consider the class of set covering games with $n \geq 2$ users. The following statements are true:
\begin{enumerate}
%
%
%
\item[(a)] For any number of rounds $k \geq 1$, the transient performance guarantees associated with the utility mechanisms $f_{MC}$ and $f_\poa$ satisfy
\begin{align}
    \pob(f_{MC}, n, k) &= 1/2, \label{eq:pobSCMC1}\\
    \pob(f_\poa, n, k) &\leq 1/2.
\end{align}
Furthermore, the marginal contribution rule achieves its asymptotic guarantees after a single round, i.e.,
\begin{equation}
    \pob(f_{MC}, n, k) = \pob(f_{MC}, n, 1).
\end{equation}
\item[(b)] For any number of rounds $k \geq 1$, the marginal contribution rule $f_{MC}$ optimizes the $\pob$, i.e.,  
\begin{equation}
    \max_{f} \pob(f, n, k) = \pob(f_{MC}, n, k) = 1/2.
\end{equation}
Alternatively, there is no utility rule $f$ with $\pob(f,n,k)>1/2$ for any $k$.  
\end{enumerate}
\end{thm}
\begin{proof}
To prove the statements outlined in items $\rm{(a)}$ and $\rm{(b)}$ of the theorem statement, we show, more generally, that for any $f$ and any $k$, the following series of inequalities is true:
\begin{align*}
    \frac{1}{2} = \pob(f_{MC}, n, 1) = \pob(f_{MC}, n, k) \geq \pob(f, n, k).
\end{align*}
The fact that $\pob(f_{MC}, n, 1) = 1/2$ is a result of Lemma \ref{lem:setpob} (stated later in the paper) for $n \geq 2$ and the definition of $f_{MC}$. For proving $\pob(f_{MC}, n, k) = \frac{1}{2}$, we first show that 
\begin{equation}
    \pob(f_{MC}, n, k) \geq \pob(f_{MC}, n, 1).
\end{equation} Let $a, a' \in \ac$ be joint actions such that $a'$ is a best response to $a$ for agent $i$. If the MC rule is used as in Equation \eqref{eq:MCdefgen}, then 
\begin{align*}
    \W(a') - \W(a) &= \W(a_i', a_{-i}) - \W(\emp_i, a_{-i}) \\
    &\qquad + \W(\emp_i, a_{-i}) - \W(a_i, a_{-i}) \\
    &= \U(a_i', a_{-i}) - \U(a_i, a_{-i}) \\
    &\geq 0.
\end{align*}
Then, after any $m \geq 1$ best response rounds, $\W(\kend(m)) \geq \W(\kend(1))$, and therefore the claim is shown. For the other direction as well as the inequality $\frac{1}{2} \geq \pob(f, n, k)$, we construct a game construction (also shown in Figure \ref{fig:2setworst}) such that for any utility rule $f$, $\pob(\G^f, k) \leq \frac{1}{2}$. The game $\G^f$ is constructed as follows with $2$ agents. Let the resource set be $\rr = \{r_1, r_2, r_3, r_4\}$  with $v_{r_1} = 1$, $v_{r_2} = 1$, $v_{r_3} = f(2)$, and $v_{r_4} = 0$. For agent $1$, let $\ac_1 = \{\emp_1, a^1_1, a^2_1\}$ with $a^1_1 = \{r_1\}$ and $a^2_1 = \{r_2\}$. For agent $2$, let $\ac_2 = \{\emp_2, a^1_2, a^2_2\}$ with $a^1_2 = \{r_3\}$ and $a^2_2 = \{r_1\}$. It can be seen in this game, if $f(2) \leq 1$, the optimal joint action is $(a^2_1, a^2_2)$ resulting in a welfare of $2$, and if $f(2) > 1$, the optimal joint action is $(a^2_1, a^1_2)$ resulting in a welfare of $1 + f(2) \geq 2$. Additionally, starting from $(\emp_1, \emp_2)$, note the best response path $((\emp_1, \emp_2), (a^1_1, \emp_2) (a^1_1, a^1_2), \dots, (a^1_1, a^1_2), (a^1_1, a^1_2), (a^1_1, a^2_2))$ exists, ending in the action $(a^1_1, a^2_2)$ with a welfare $1$. Since $(a^1_1, a^1_2)$ is a Nash equilibrium, both agents stay playing their first action $a^1$ for $k-1$ rounds and agent $2$ only switches to $a^2_2$ in the $k$-round. Therefore $\pob(\G^f, k) \leq \frac{1}{2}$ for this game. This game construction $\G^f$ can be extended to $n > 2$ by fixing the actions of agents $i \geq 2$ to be $\ac_i = \{\emp_i, a^1_i\}$, where the only nonempty action available to them is $a^1_i = \{r_4\}$ that selects the $0$ value resource.

Since $\G^f$ may not be the worst case game, $\pob(f, n, k) \leq \pob(\G^f, k) \leq \frac{1}{2}$ and the final claim is shown.
\end{proof}

\begin{figure}[ht]
    \centering
    \includegraphics[width=220pt]{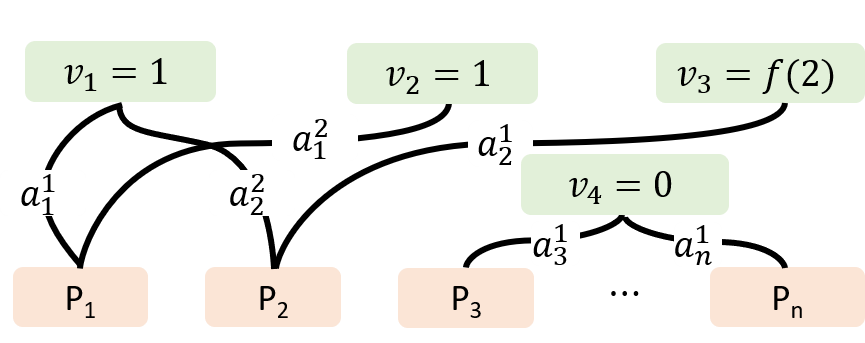}
    \caption{Under this $n$-agent game construction $\G^f$, outlined in Theorem \ref{thm:MCopt}'s proof, $\pob(\G^f, k) \leq \frac{1}{2}$ for any mechanism $f$. There are four resources $\rr = \{r_1, r_2, r_3, r_4\}$ with values $\{1, 1, f(2), 0\}$ respectively. Under the allocation $(a^2_1, a^2_2)$, the resources $r_1$, $r_2$, and $r_4$ are selected to produce a welfare of $2$. However for any number of rounds $k \geq 1$, there exists a best response path that ends in an allocation that selects only the resources $r_1$ and $r_4$ to produce a welfare of $1$. Then $\pob(f, n, k) \leq \pob(\G^f, k) \leq \frac{1}{2}$.}
    \label{fig:2setworst}
\end{figure}


Observe that Theorem \ref{thm:MCopt} answers Questions (a)--(c) posed in the previous section. Specifically, we have shown that for any number of rounds, the performance guarantees of $f_{\poa}$ \emph{never} overcome the performance guarantees of $f_{MC}$. Remarkably, for any fixed number of rounds $k$, $\pob(f_{MC}, k)$ remains greater than $\pob(f_{\poa}, k)$ and the advantages of using $f_{\poa}$ only appear when the system settles to a Nash equilibrium as $k\to\infty$. Furthermore, the marginal contribution rule is in fact the utility mechanism that optimizes the price of best response and that this optimal efficiency guarantee is achieved after $1$ round of best response. Thus, the transient guarantees of $f_{MC}$ are quite strong in comparison to $f_\poa$. Note that this is in stark contrast with their price of anarchy values, where $f_{\poa}$ significantly outperforms $f_{MC}$ \cite{gairing2009covering,ramaswamy2019multiagent}. Clearly, there exists a trade-off between short-term and long-term performance guarantees in this setting.

In our next theorem, we characterize the trade-off between short-term and long-term performance guarantees by providing explicit expressions for the Pareto curve between $\pob(f, 1)$ and $\poa$ and the corresponding Pareto-optimal utility mechanisms. Surprisingly, we show that $\pob(f_{\poa}, 1)$ can be arbitrarily bad.

\begin{thm}
\label{thm:poapobtradeoff}
For a fixed $\poa(f) = C \in [\frac{1}{2}, 1 - \frac{1}{e}]$, the optimal $\pob(f, 1)$ is
\begin{equation}
\label{eq:scpoba}
\small \max_f \pob(f, 1) = \Bigg[ \sum^\infty_{j=0} \max \{ j! (1 - \frac{1 - C}{C} \sum^j_{\tau=1}{\frac{1}{\tau!}}) , 0\}+1 \Bigg]^{-1}.
\end{equation}
Furthermore, if $\poa(f) =  1 - \frac{1}{e}$ for a given $f$, then $\pob(f, 1)$ must be equal to $0$.
\end{thm}

The price of anarchy is only taken on the interval $[\frac{1}{2}, 1 - \frac{1}{e}]$, since the greatest achievable $\poa(f)$ is $1 - \frac{1}{e}$ \cite{gairing2009covering} and the price of anarchy of $f_{MC}$ is $\frac{1}{2}$ (which attains the optimal $\pob(f, 1)$ as shown in Theorem \ref{thm:MCopt}). The corresponding Pareto-optimal utility rules are given in Lemma \ref{lem:implicitfboundar}. The optimal trade-off curve between $\poa$ and $\pob$ is plotted in Figure \ref{fig:poapobcurve}.

\begin{figure}[ht]
    \centering
    \includegraphics[width=220pt]{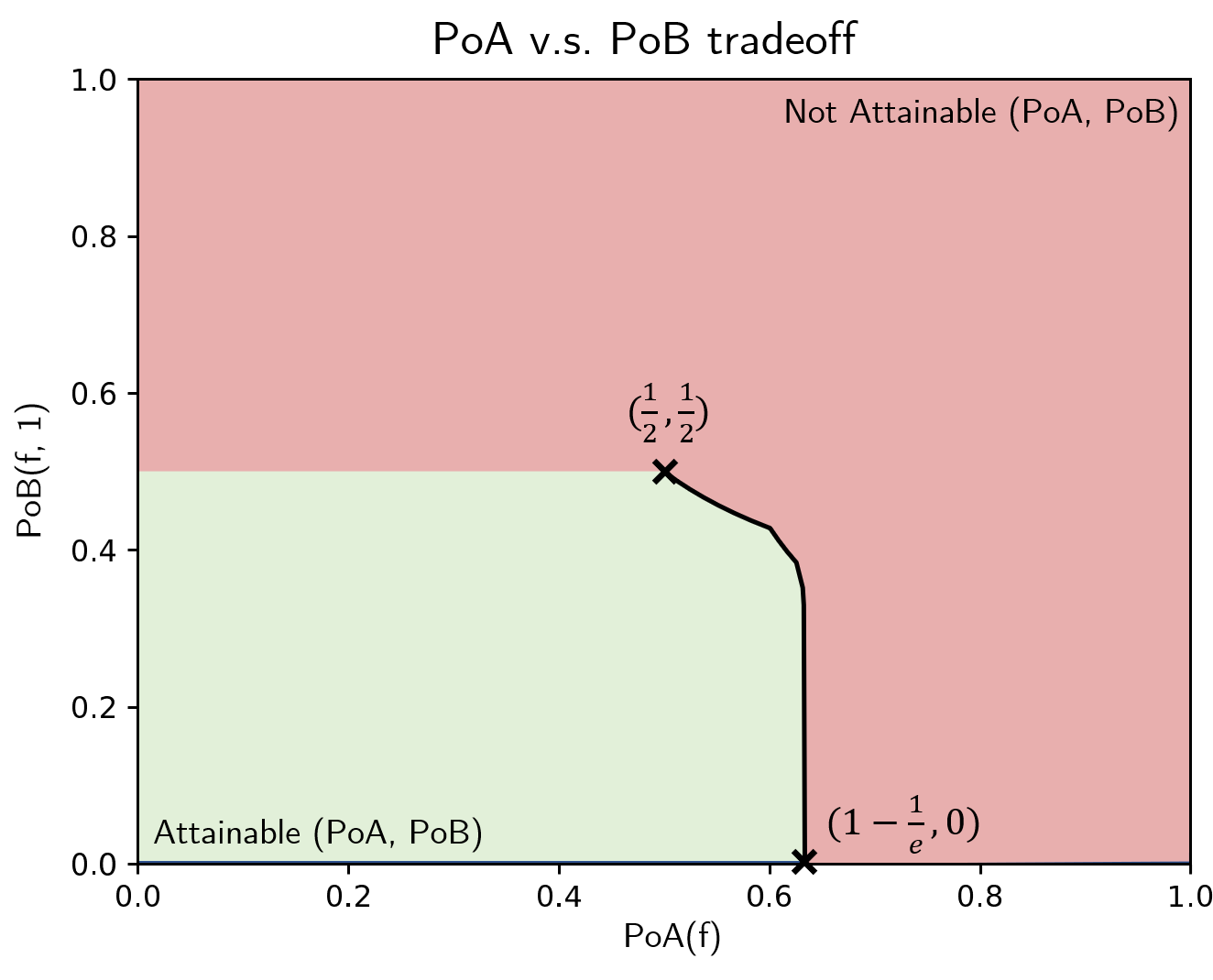}
    \caption{All possible pairs $(\poa(f), \pob(f, 1)) \in [0, 1] \times [0, 1]$ are depicted in this figure. The attainable set where there exists a mechanism $f$ that can achieve either a greater $\pob(f, 1)$ or $\poa(f)$ is depicted in green and the complement is depicted in red. The Pareto curve that results from the optimal trade-off described in Theorem \ref{thm:poapobtradeoff} is depicted by the black curve. We observe that there is an extreme drop in transient performance to $\pob(f, 1)=0$ when using any utility rule that optimizes $\poa(f)$.}
    \label{fig:poapobcurve}
\end{figure}

Theorem \ref{thm:poapobtradeoff} highlights the disparity between the efficiency guarantees of the asymptotic and transient values, quantified by the pair of achievable $\poa$ and $\pob$. While there is a decline in $\poa$ guarantees when using the MC rule, the deterioration in $\pob$ that results from using the mechanism $f_{\poa}$ that optimizes $\poa$ is much more extreme. In fact, the transient performance of $f_{\poa}$ can be arbitrarily bad. This is quite surprising, it appears that a utility rule's asymptotic performance guarantees do not reflect on its transient guarantees. This extreme trade-off has significant consequences on distributed design, especially when the short-term performance of the multi-agent system is critical, and prompts a more careful interpretation of the asymptotic results in such scenarios.

\section{Proof of Theorem \ref{thm:poapobtradeoff}}
\label{sec:proofthm}
In this section, we prove the result in Theorem \ref{thm:poapobtradeoff}. First, we characterize the $1$-round price of best response for any utility rule $f$ in the next lemma. The proof of this lemma is based on a worst case game construction using a smoothness argument. The worst case game construction is well structured and is shown in Figure \ref{fig:worst1}.

\begin{lemma}
\label{lem:setpob}
Given a utility rule $f$ and $n$ number of agents, the $1$-round price of best response is
\begin{equation}
    \label{eq:pobwscf}
    \frac{1}{\pob(f, n, 1)} = \Big( \sum_{j=1}^{n}{f(j)} - \min_j{f(j)} \Big) / f(1) + 1.
\end{equation}
\end{lemma}
\begin{proof}
To characterize $\pob(f, n, 1)$, we first formulate it as an optimization program, where we search over all possible set covering games with a fixed $f$ and $n$ for the worst $1$-round outcome. We give a proof sketch of the correctness of the proposed optimization program in \eqref{eq:dualpobsc}. Without loss of generality, we can only consider games where each agent $i$ has $2$ actions $\opt_i$ and $\abr_i$, in which $\opt$ is the action that optimizes the welfare $\W$ and $\abr$ is the action that results from a $1$-round best response. Informally, the optimization program involves searching for a set covering game that maximizes $\W(\opt)$ with the constraints $\W(\abr) = 1$ and that $\abr$ is indeed the state that is achieved after a $1$-round best response. Since the welfare $\W(a)$ as well as $\U(a)$ can be described fully by $v_r$ and $f$, each game is parametrized through the resource values $v_r$ from a common resource set $\rr^*$. Let 
\begin{equation}
    \babr(r) = \{i \in \p : r \in \abr_i \}
\end{equation}
denote the set of agents that select the resource $r$ in $\abr$. We define $\baopt(r)$ similarly. Additionally, we can define $\babr^{<i}(r) = \{1 \leq j < i  : r \in \abr_j \}$ to denote the set of agents $\{1, \dots, i-1\}$ that select the resource $r$ in in $\abr_i$. The common resource set $\rr^*$ is comprised of all the resources $r$ that have unique signature from $\babr(r)$ and $\baopt(r)$. Taking the dual of the program in the previous outline results in the following program.
\begin{align}
    \pob(f, n, 1)^{-1} &= \min_{\{\lambda_i\}_{i \in \p}, \mu} \mu \text{ s.t.} \label{eq:dualpobsc}\\
    \label{eq:dualpobcson}
    \mu \min(|\ba_{\mathrm{br}}(r)|, 1) &\geq \min(|\ba_{\mathrm{opt}}(r)|, 1) + \\ 
    \medmath{\sum_{i \in \p} \lambda_i \Big[ \big(\mathbf{1}_{\babr(r)}(i)} & \medmath{- \mathbf{1}_{\baopt(r)}(i)\big) f(|\babr^{<i}(r)| + 1) \Big] \geq 0} \ \forall r \in \rr^* \nonumber
\end{align}
Now we reduce the constraints of this program to give a closed form expression of $\pobset$.
We first examine the constraints in \eqref{eq:dualpobcson} associated with $r \in \rr^*$ with $\babr(r) = \varnothing$. Simplifying the corresponding constraints gives
\begin{equation}
\label{eq:dualsetnullbr}
\sum_{j \in P} \lambda_j f(1) \geq 1
\end{equation}
for any subset of agents $P \subset \p$. Only considering $P = \{j\}$ gives the set of binding constraints $\lambda_j \geq 1/f(1)$ for any agent $j$.

Now we consider the converse constraints in \eqref{eq:dualpobcson} where $\babr(r) \neq \varnothing$. For this set of constraints, the binding constraint is postulated to be 
\begin{equation}
\label{eq:dualsetnotnullbr}
\mu \geq \Big( \sum_{j=1}^{n}{ \lambda_j f(j)} - \min_j{f(j)} \Big) + 1
\end{equation}
corresponding to the resource $\bar{r}$ with $\babr(\bar{r})  = \p$ and $\baopt(\bar{r}) = \{\arg \min_j f(j)\}$. Under this assumption, the optimal dual variables $\lambda_j$ are $\lambda_j = 1/f(1) \ \forall j$. This follows from the constraint in \eqref{eq:dualsetnullbr} and the fact that decreasing any $\lambda_j$ results in a lower binding constraint in \eqref{eq:dualsetnotnullbr}.

Now we can verify that the constraint in \eqref{eq:dualsetnotnullbr} with $\babr(\bar{r})  = \p$ and $\baopt(\bar{r}) = \{\arg \min_j f(j)\}$ is indeed the binding constraint under the dual variables  $\lambda_j = 1/f(1)$. The resulting dual program assuming $\babr(r) \neq \varnothing$ is
\begin{align}
    &\min_{\mu} \mu \text{ s.t.} \\
    \label{eq:dualbrgeq1}
    \mu &\geq \sum_j \big(\mathbf{1}_{\babr(r)}(j) - \mathbf{1}_{\baopt(r)}(j) \big) \\
    &f(|\babr^{<j}(r)|+1)/f(1) + \min (|\baopt(r)|, 1) \ \forall r. \nonumber
\end{align}
Since any valid utility rule has $f(j) \geq 0$ for any $j$, the terms in \eqref{eq:dualbrgeq1} is maximized if $\mathbf{1}_{\babr(r)}(j) = 1$ for all $j$ which implies that $\babr(r) = \p$. Also note that $|\baopt(r)| \leq 1$, as 
\begin{equation}
    \label{eq:dualoptcons}
    \min (|\baopt(r)|, 1) - \sum_j \mathbf{1}_{\baopt(r)}(j) \big) f(|\babr^{<j}(r)|+1)/f(1)
\end{equation}
is less binding if $|\baopt(r)| > 1$. Since $\babr(r) = \p$ for the binding constraint, if $|\baopt(r)| = 1$, the terms in \eqref{eq:dualoptcons} reduces to $1 - f(j)/f(1)$ for some $j$. Similarly, if $|\baopt(r)| = 0$, the terms in \eqref{eq:dualoptcons} reduces to $0$. Note that $ 1 - f(1)/f(1) \geq 0$, and so $1 - \min_j \{f(j)\}/f(1) \geq 0$ as well. Thus $|\baopt(r)| = 1$ is binding. Taking any $\baopt(\bar{r}) \in \arg \min_j f(j)$ attains the maximum value for $1 - \min_j \{f(j)\}/f(1)$, and so we have shown the claim that the binding constraint for $\mu$ is in equation \eqref{eq:dualsetnotnullbr}. Taking the binding constraint with equality gives
\begin{equation*}
    \mu = \Big( \sum_{j=1}^{n}{f(j)}/f(1) - \min_j{f(j)} \Big) + 1,
\end{equation*}
giving the expression for $\pob$ as stated.
\end{proof}

\begin{figure}[ht]
    \centering
    \includegraphics[width=220pt]{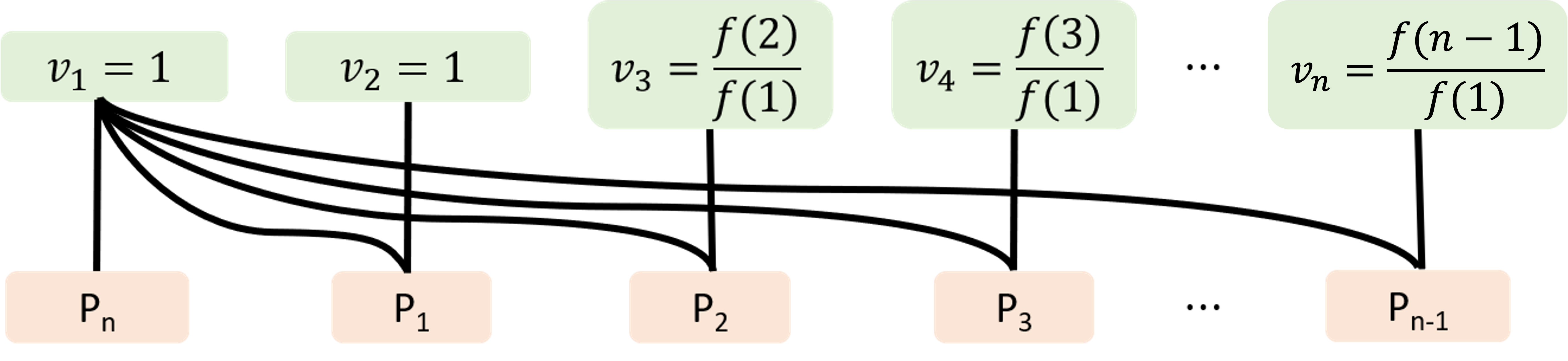}
    \caption{In this worst case game example, we assume that $f(n) = \min_j{f(j)}$. Then in a $1$-round best response, each agent can choose $r_1$ with $v_1 = 1$ successively, resulting in a welfare of $1$. When agent $i < n$ moves, note that if the previous $i-1$ agents have selected $r_1$ , agent $i$ is indifferent towards selecting the $r_1$ or $r_{i+1}$ resource. Agent $n$ can only select the resource $r_1$. The optimal allocation is when all agents select different resources, resulting in a welfare that is reported in Equation \eqref{eq:pobwscf}.}
    \label{fig:worst1}
\end{figure}

To characterize the trade-off, we now provide an explicit expression of Pareto optimal utility rules, i.e., utility rules $f$ that satisfy either $\pob(f,1)\geq\pob(f',1)$ or $\poa(f)\geq\poa(f')$ for all $f'\neq f$.

\begin{lemma} \label{lem:implicitfboundar}
For a given $\ipoa \geq 0$, a utility rule $f$ that satisfies $\poa(f)\geq 1/(1+\ipoa)$ while maximizing $\pob(f,1)$ is defined as in the following recursive formula:
\begin{align}
    f^{\ipoa}(1) &= 1 \\
    \label{eq:recurf}
    f^{\ipoa}(j+1) &= \max\{j f^{\ipoa}(j) - \ipoa, 0\}.
\end{align}
\end{lemma}
\begin{proof}
According to a modified version of Theorem $2$ in \cite{paccagnan2018utility}, the price of anarchy ($\poa$) can be written as 
\begin{align}
\label{eq:poasetgen}
    \frac{1}{\poa(f)} &= 1 + \Big( \max_{1 \leq j \leq n - 1} \{(j+1) f(j+1) - f(1), \nonumber \\
    &jf(j)-f(j+1), jf(j+1)\} \Big)/f(1).
\end{align}

We first show that if $f$ is Pareto optimal, then it must also be non-increasing. Otherwise, we show that another $f$ exists that achieves at least the same $\pob$, but a higher $\poa$, contradicting our assumption that $f$ is Pareto optimal. Assume, by contradiction, that there exists $f$ that is Pareto optimal and decreasing, i.e., there exists a $j' \geq 1$, in which $f(j') < f(j'+1)$. Notice that switching the value $f(j')$ with $f(j'+1)$ results in an unchanged $\pob$ according to \eqref{eq:pobwscf} in Lemma \ref{lem:setpob} if $j' > 1$. We show that $f'$ with the values switched has a higher $\poa$ than $f$.

For any $1 \leq j' \leq n - 1$, the expressions from \eqref{eq:poasetgen} that include $f(j')$ or $f(j'+1)$ are
\begin{align*}
    (j' f(j') - f(1)))/f(1) + 1, \\
    ((j'+1) f(j'+1) - f(1))/f(1) + 1, \\
    ((j'-1) f(j'-1) - f(j'))/f(1) + 1, \\
    (j' f(j') - f(j'+1))/f(1) + 1, \\
    ((j'+1) f(j'+1) - f(j'+2))/f(1) + 1, \\
    ((j'-1) f(j'))/f(1) + 1, \\
    (j' f(j'+1))/f(1) + 1,
\end{align*}
After switching, the relevant expressions for $f'$ are
\begin{align*}
    (j' f(j'+1) - f(1)))/f(1) + 1, \\
    ((j'+1) f(j') - f(1))/f(1) + 1, \\
    ((j'-1) f(j'-1) - f(j'+1))/f(1) + 1, \\
    (j' f(j'+1) - f(j'))/f(1) + 1, \\
    ((j'+1) f(j') - f(j'+2))/f(1) + 1, \\
    ((j'-1) f(j'+1))/f(1) + 1, \\ 
    (j' f(j'))/f(1) + 1.
\end{align*}
Since $f(j') < f(j'+1)$, switching the values results in a strictly looser set of constraints, and the value of the binding constraint in \eqref{eq:poasetgen} for $f'$ is at most the value of the binding constraint for $f$. Therefore $\poa(f) \leq \poa(f')$. Note that if $j'= 1$, then $\pob(f') > \pob(f)$ as well. This contradicts our assumption that $f$ is Pareto optimal.

Now we restrict our focus to non-increasing $f$ with $f(1) = 1$. Observe that for any $f$, the transformation $f'(j) = f(j)/C$, $C>0$, doesn't change the price of best response and price of anarchy. Thus, we can take $f(1) = 1$ without loss of generality. Given a $f$ that satisfies the above constraints, the price of anarchy is
\begin{equation}
    \frac{1}{\poa(f)} = 1 + \\ \max_{1 \leq j \leq n - 1} \{jf(j)-f(j+1), (n-1)f(n)\},
\end{equation}
as detailed in Corollary $2$ in \cite{paccagnan2018utility}. Let 
\begin{equation}
\label{eq:chiconst}
    \ipoa_f = \max_{1 \leq j \leq n - 1} \{jf(j)-f(j+1), (n-1)f(n)\}.
\end{equation} 
For $f$ to be Pareto optimal, we claim that $\ipoa_{f} = jf(j)-f(j+1)$ must hold for all $j$. Consider any other $f'$ with $\ipoa_f = \ipoa_{f'}$. It follows that $\poa(f) = \poa(f') = 1/(1+\ipoa_f)$. By induction, we show that $f(j) \leq f'(j)$ for all $j$. The base case is satisfied, as $1 = f(1) \leq f'(1) = 1$. Under the assumption $f(j) \leq f'(j)$, we also have that 
\begin{equation}
    jf(j) - \ipoa_f = f(j+1) \leq f'(j+1) = jf(j) - \ipoa^j_f,
\end{equation}
where $\ipoa^j_f = jf(j) - f(j+1) \leq \ipoa_f$ by definition in \eqref{eq:chiconst}, and so $f(j) \leq f'(j)$ for all $j$. Therefore the summation $\sum_{j=1}^{n}{f(j)} - \min_j{f(j)}$ in \eqref{eq:pobwscf} is diminished and $\pob(f) \geq \pob(f')$, proving our claim. As $f$ must satisfy $f(j) \geq 0$ for all $j$ to be a valid utility rule, $f(j+1)$ is set to be $\max \{jf(j) - \ipoa, 0\}$. Then we get the recursive definition for the maximal $f^{\ipoa}$. Finally, we note that for infinite $n$, $\ipoa \leq \frac{1}{e-1}$ is not achievable, as shown in \cite{gairing2009covering}.
\end{proof}

With the two previous lemmas, we can move to showing the result detailed in Theorem \ref{thm:poapobtradeoff}.

\begin{proof}[Proof of Theorem \ref{thm:poapobtradeoff}]
We first characterize a closed form expression of the maximal utility rule $f^{\ipoa}$, which is given in Lemma \ref{lem:implicitfboundar}. We fix $\ipoa$ so that $\poa(f^{\ipoa}) = \frac{1}{\ipoa+1}$ = $C$. To calculate the expression for $f^{\ipoa}$ for a given $\ipoa$, a corresponding time varying discrete time system to \eqref{eq:recurf} is constructed as follows.
\begin{align}
    x[m+1] &= m x[m] - u[m], \\
    y[m] &= \max \{x[m], 0\}, \\
    u[m] &= \ipoa, \\
    x[1] &= 1,
\end{align}
where $y[m]$ corresponds to $f^{\ipoa}(j)$. Solving for the explicit expression for $y[m]$ gives
\begin{align}
    y[1] &= 1 \\
    y[m] &= \max \Big[ \prod_{\ell=1}^{m-1} \ell - \ipoa \big( \sum_{\tau=1}^{m-2} \prod_{\ell=\tau+1}^{m-1} \ell \big) - \ipoa , 0 \Big] \ \ m > 1.
\end{align}
Simplifying the expression and substituting for $f^\ipoa(j)$ gives 
\begin{equation}
    f(j) = \max \Big[ (j-1)!(1 - \ipoa\sum_{\tau=1}^{j-1} \frac{1}{\tau!} \big), 0 \Big] \ \ \ j \geq 1.
\end{equation}

Substituting the expression for the maximal $f$ into \eqref{eq:pobwscf} gives the $\pob$. Notice that for $\ipoa \geq \frac{1}{e-1}$, $\lim_{j \to \infty} f(j) = 0$, and therefore $\min_j{f(j)} = 0$. Shifting variables $j' = j + 1$, we get the statement in \eqref{eq:scpoba}.

Finally, we prove that if $\poa(f) =  1 - \frac{1}{e}$ for a given $f$, then $\pob(f, 1)$ must be equal to $0$. According to \eqref{eq:scpoba}, among utility rules $f$ with $\poa(f) = 1 - \frac{1}{e}$, the best achievable $\pob(f, 1)$ can be written as
\begin{align*}
  \max_f \pob(f, 1) & =  \Big( \sum_{j=0}^{\infty}{\frac{j! \sum_{\tau = j+1}^{\infty}{\frac{1}{\tau!}}}{e-1}} + 1 \Big)^{-1} \\
& \leq \Big( \frac{1}{e-1} \sum_{j=0}^{\infty}{\frac{1}{j+1}} + 1 \Big)^{-1} \leq 0
\end{align*}
Since $0 \leq \pob( f, 1) \leq   \max_f \pob(f, 1) \leq 0$, it follows that $\pob(f, 1)=0$.
\end{proof}

\section{Simulations}
\label{sec:sim}
We analyze our theoretical results empirically through a Monte Carlo simulation, presented in Figure \ref{fig:compeff}. We simulate $200$ instances where a random game $\G$ is generated and the welfare along a $4$-round best response path is recorded. Each game $\G$ has $10$ agents and two sets of $10$ resources, where the values $v_r$ are drawn uniformly from the continuous interval $[0, 1]$. Each agent in the game has two singleton actions, where in one action the agent selects a resource in the first set and in the other action the agent selects a resource in the second set, as seen in Figure \ref{fig:simworst}. The two resources in each agent's actions are selected uniformly, at random. These parameters were chosen to induce a natural set of game circumstances that have a rich best response behavior. Additionally, the game is either endowed with the utility rule $f_{MC}$ that optimizes $\pob(f, 1)$ or $f_{\poa}$ that optimizes $\poa(f)$. Then the best response round algorithm is performed, where each agent successively performs a best response. The ratio of the welfare achieved at each time step following the two utility rules $f_{MC}$, $f_{\poa}$ is displayed in Figure \ref{fig:compeff}.

\begin{figure}[ht]
    \centering
    \includegraphics[width=240pt]{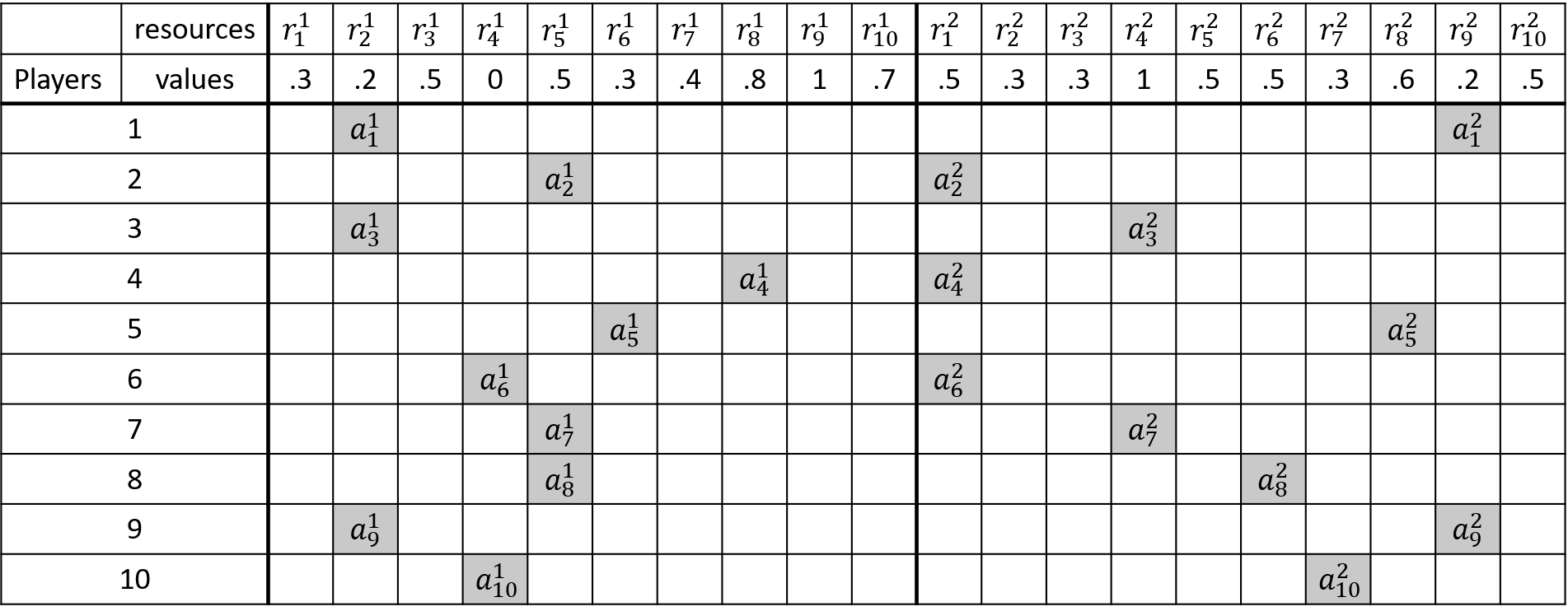}
    \caption{An example of a game that was generated within the Monte Carlo simulation described in Section \ref{sec:sim}. The highlighted boxes represent the possible single resource selections that each agent can select. The two sets of resources are divided by a bold line. For each game generated, the welfare ratio between the best response sequences corresponding with the utility rules $f_{MC}$ and $f_{\poa}$ was recorded at each time step.}
    \label{fig:simworst}
\end{figure}

\begin{figure}[ht]
    \centering
    \includegraphics[width=270pt]{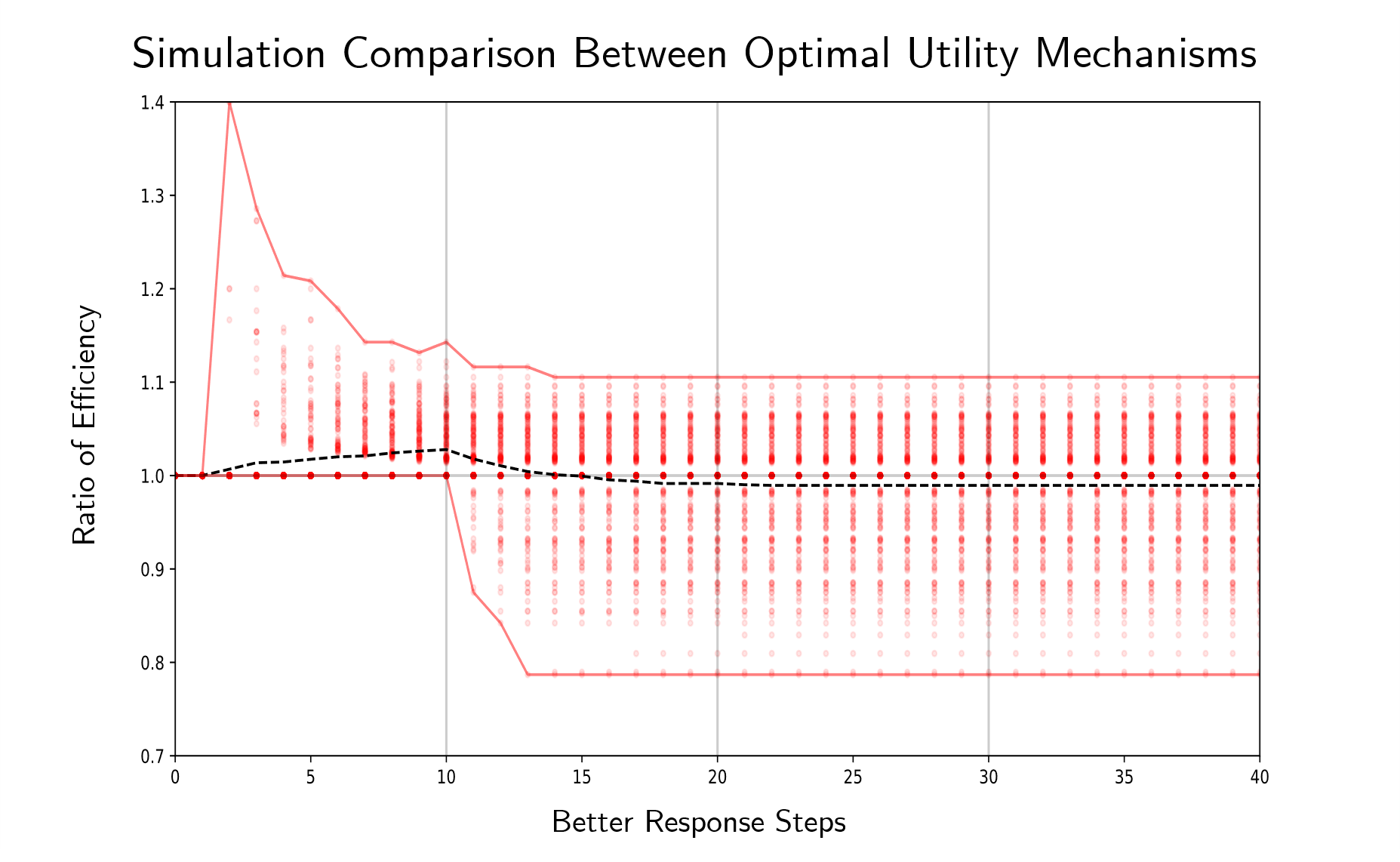}
    \caption{A Monte Carlo simulation with $10$ agents and $2$ set of resources was conducted, where the utility rules $f_{MC}$ and $f_{\poa}$ were used in separate runs. All the agents are randomly assigned a single resource from each set in their $2$ actions. The ratio $\W_{MC}(a)/\W_{\poa}(a)$ (from using either of the utility rules) is plotted at each time step along the two better response paths. The average ratio across all runs is displayed in black and the minimum and maximum ratios are outlined with a red line. While on average, the performance from using $f_{MC}$ is greater than $f_{\poa}$ for $1$ round, this performance difference is recuperated after running the best response algorithm for more steps.}
    \label{fig:compeff}
\end{figure}

Several significant observations can be made from Figure \ref{fig:compeff}. Observe that after one round (after $30$ best response steps), the performance of $f_{MC}$ exceeds the performance $f_{\poa}$ on average. This is indicative of the trade-off shown in Theorem \ref{thm:poapobtradeoff}. Furthermore, running the best-response algorithm for additional steps recovers the performance difference between $f_{MC}$ and $f_{\poa}$, with $f_{\poa}$ having a slightly better end performance on average. These results support our claim that optimizing for asymptotic results may induce sub-optimal system behavior in the short term, and those results must be carefully interpreted for application in multi-agent scenarios. By analyzing the performance directly along these dynamics, we see that a richer story emerges in both the theoretical and the empirical results.

\section{Conclusion}
\label{sec:conc}
There are many works that focus on efficiency guarantees for the asymptotic behavior of multi-agent systems under the game-theoretic notion of Nash equilibrium. However, much less is understood about the efficiency guarantees in the transient behavior of such systems. In this paper, we examine a natural class of best response dynamics and characterize the efficiency along the resulting transient states. We characterize a utility design that performs optimally in the transient for the class of set covering games. Furthermore, we complement these results by providing an explicit characterization of the efficiency trade-off between $1$-round best response and Nash equilibrium that can result from a possible utility design. From these results, we observe that there is a significant misalignment between the optimization of transient and asymptotic behavior for these systems. Interesting future directions include extending our results to a more general class of games and considering other classes of learning dynamics from the literature.

\bibliographystyle{ieeetr}
\bibliography{references.bib}
\end{document}